\documentclass[12pt]{article}
\usepackage{times}

\usepackage{amsthm,amsfonts,amssymb,amsmath} 

\newcommand{\R}{\mathbb{R}}
\newcommand{\dd}{\mathrm{d}}
\newcommand{\pa}{\partial}

\newcommand{\wt}{\widetilde}

\usepackage[T1]{fontenc}
\usepackage[cp1250]{inputenc}
\newcommand{\be}{\begin{equation}}
\newcommand{\ee}{\end{equation}}

\newcommand{\ba}{\begin{eqnarray}}
\newcommand{\ea}{\end{eqnarray}}

\usepackage{cite}
\usepackage{color}
\newtheorem{thm}{Theorem}[section]
\newtheorem{defi}[thm]{Definition}
\newtheorem{prop}[thm]{Proposition}
\newtheorem{lem}[thm]{Lemma}

\newcommand{\p}{\partial}

\newtheorem{cor}{Corollary}

\begin{document}
\centerline{\textbf{New solutions of initial conditions in general relativity}}

\date{\today}

\null

\centerline{ J Tafel$^1$, M J\'{o}\'{z}wikowski$^2$} 
\noindent
{$^1$ Institute of Theoretical Physics, Warsaw University, Ho\.{z}a 69, 00-681 Warsaw, Poland
  ( tafel@fuw.edu.pl)}\\
{$^2$ Institute of Mathematics, Polish Academy of Sciences, \'Sniadeckich 8, 00-956 Warsaw, Poland (mjozwikowski@gmail.com).

\begin{abstract}
We find new classes of exact  solutions of the initial momentum constraint  for vacuum Einstein's equations.  Considered data are either invariant under a continuous symmetry  or they are assumed to have the exterior curvature tensor of a simple form. In general the mean curvature $H$ is non-constant and  $g$ is not conformally flat. In the generic case with the symmetry we obtain general solution in an explicit form. In other cases solutions are given up to quadrature. We also find a class of explicit solutions without symmetries which generalizes data induced by the Kerr metric or other metrics related to the Ernst equation.  The conformal method of Lichnerowicz, Choquet-Bruhat and York is used to prove solvability of the Hamiltonian constraint if $H$ vanishes. Existence of marginally outer trapped surfaces in initial manifold is discussed.
\end{abstract}

\null


\null


\null

\section{Introduction}

Initial data for the vacuum Einstein equations consists of a 3-dimensional manifold $S$ together with a Riemannian metric $g_{ij}$ and a symmetric tensor $K_{ij}$. They have to satisfy the following constraint equations 
\begin{align}
\nabla_i\big(K^{ij}-g^{ij}H\big)&=0 \label{eqn:vector}\\
  R+H^2-K_{ij}K^{ij}&=0,\label{eqn:scalar}
\end{align} 
where  $\nabla_i$ are covariant derivatives, $R$ is the Ricci scalar of $g_{ij}$ and $H=K^i_{\ i}$. Tensors $g_{ij}$ and  $K_{ij}$ are interpreted, respectively, as the induced metric and the external curvature of $S$ embedded in  4-dimensional spacetime $M$ with metric $g^{(4)}$
 \be
g _{ij}dx^i dx^j=e^*g^{(4)}\ ,\ \ \ 2K_{ij}dx^i dx^j=e^*L_kg^{(4)}\label{-2}
\ee
($L_k$ is the Lie derivative along the  future  unit normal vector of $S$ and $e^*$ denotes  the pullback under an embedding $e:S\rightarrow M$).

In order to analyze  constraints (\ref{eqn:vector})  and (\ref{eqn:scalar}) one can employ the method, originated by Arnowitt, Deser and Misner \cite{adm}, of transverse-traceless (TT) decomposition of $K_{ij}$ into its trace, a longitudinal part and a divergence free traceless tensor  \cite{York_73}. Combining this decomposition with a conformal transformation of the initial data yields a system of differential equations equivalent to (\ref{eqn:vector}) and (\ref{eqn:scalar}). The main advantage of this approach (known as the conformal method) is that in the new system one can distinguish dynamical variables from free functions (note that system  \eqref{eqn:vector}--\eqref{eqn:scalar} is under-determined). In this approach initial data sets with the vanishing  mean curvature $H$  are constructed by means of
 the conformal transformation
\begin{equation}\label{eqn:conf}
g'_{ij}=\psi^4 g_{ij},\quad K'_{ij}=\psi^{-2} K_{ij}\ ,\ \ \psi>0\ .
\end{equation}
 If $H=0$ and tensor $K^{}_{ij}$  satisfies the momentum constraint (\ref{eqn:vector}) then this transformation yields another traceless ($H'=0$) solution of this constraint. 
The new fields satisfy the Hamiltonian constraint \eqref{eqn:scalar} iff
\begin{equation}\label{eqn:scalar1}
\bigtriangleup\psi=\frac 18 R\psi-\frac 18  K_{ij} K^{ij}\psi^{-7}\ ,
\end{equation}
where $\nabla_i$ denotes the covariant derivative and $\bigtriangleup$  is the covariant Laplace operator for  metric $g_{ij}$. The pair $(g'_{ij},K'_{ij})$ constitutes unconstrained  initial data for the vacuum Einstein equations.

All known exact solutions of the constraint equations are based on the above scheme. Among them there are  solutions of Brill and Lindquist \cite{Br_Lindq_63}, Bowen and York \cite{BY_80} and Brandt and Brugmann \cite{Br-Br_97}. In all of them metric $g_{ij}$ is flat, hence the momentum constraint can be easily solved. In general, equation \eqref{eqn:scalar1} is not solvable analytically. The existence of its solutions  was determined in cases where $(S,g_{ij})$ is closed \cite{ChB_York_79,Isenberg_95}, asymptotically flat \cite{Cantor_77,Br_Cant_81, ma} or asymptotically hyperbolic \cite{And_Chr_96, And_Chr_Fri_92}. Also particular results are known for $(S,g_{ij})$ asymptotically flat with an interior boundary \cite{ma,  Dain_04}.  A detailed discussion of the problem can be found in \cite{Bartnik_Is_04}.

In this paper we construct new solutions of the momentum constraint   with constant or non-constant mean curvature without assuming that the initial metric is conformally flat. For brevity we write  (\ref{eqn:vector}) in the form 
\be
\nabla_i T^{ij}=0\label{4}
\ee
where $T^{ij}=K^{ij}-Hg^{ij}$. 
In Section 2 we assume that data  are preserved by a continuous symmetry (which is the Killing symmetry of the corresponding 4-dimensional metric).  In generic case all the solutions of the momentum constraint are found explicitly. They depend on  6 free functions (including 3 degrees of freedom of the metric tensor), which cannot be gauged away. They generalize solutions found in \cite{bp,d} and contain  data induced by  the Kerr metric (and other stationary axially symmetric metrics) what is not possible in the case of conformally flat initial metrics \cite{Gar_Pri_00}.  In Section 3 we consider non-symmetric data but we make an algebraic  assumption about  $T_{ij}$. We obtain a particular class of solutions which depend on 3 free functions of 3 coordinates (the same number as in the case of the flat metric). It also contains data for axially symmetric vacuum metrics  related to  solutions of the Ernst equation.

In principle,  free functions in the data can be used to fulfill the Hamiltonian constraint. We shortly discuss this possibility in the begining of Section 4, however we don't have rigorous results in this direction. If  $T^i_{\ i}=0$ then $H=0$ and  one can use the conformal method  to construct implicitly  data which satisfy  all initial constraints (in practical applications the Hamiltonian constraint has to be solved numerically). In Section 4  we prove the existence of solutions of the Lichnerowicz equation (\ref{eqn:scalar1}) following results of Maxwell \cite{ma} for asymptotically flat data.  Maxwell's approach for manifolds  with a boundary  is used in Section 5 in order to construct initial data with  marginally outer trapped surfaces.

\section{Symmetric solutions of the momentum constraint}\label{ssec:v_nontwist1}
In this section we assume that  initial data is preserved by a vector field $v$.  This field  extends to  the  Killing vector of the corresponding 4-dimensional  metric $g^{(4)}$. A class of data of this kind was  constructed in \cite{bp,d} with the exterior curvature equivalent to (\ref{28}) and  (\ref{45}). Its generalization was recently considered in \cite{cm}.

 The following propositions describe general  solution of  the momentum constraint with the symmetry $v$. We use nonholonomic basis  written in coordinates $x^i=x^a,\varphi$, with $a=1,2$, such that  $v=\p_{\varphi}$.
\begin{prop}\label{p1}
Let 
\begin{equation}
g=g_{ab}dx^adx^b+\alpha^2 (d\varphi+\beta_a dx^a)^2\label{27}
\end{equation}
 and  components of metric and the exterior curvature  be independent of $\varphi$. 
Then
\begin{itemize}
 \item 
In the basis $\theta^a=dx^a$, $\theta^3=d\varphi+\beta_a dx^a$ the momentum constraint is equivalent to an explicit formula for $T^{\ a}_3$
\begin{equation}
T^{\ a}_3=\alpha^{-1}\eta^{ab}\omega _{,b}\label{28}
\end{equation}
and  equation
\begin{equation}
(\alpha T^{\ b}_{a})_{|b}= T^{\ 3}_3\alpha_{,a}+\lambda \omega_{,a}\ ,\ \ \lambda=\eta^{ab}\beta_{a,b}\ ,\label{27b}
\end{equation}
 where  the covariant derivative $_{|b}$ and the Levi-Civita tensor $\eta^{ab}$ correspond to  $g_{ab}$ and $\omega$ is a function of $x^a$.
\item
In complex coordinates $x^a=\zeta,\bar\zeta$ such that $g=\gamma^2d\zeta d\bar\zeta+\alpha^2 (\theta^3)^2$  equation (\ref{27b}) reads
\be
2(\alpha T_{\bar\zeta\bar\zeta})_{,\zeta}=\gamma^2T^{\ 3}_{3}\alpha_{,\bar\zeta}-\gamma^2(\alpha T_{\zeta}^{\ \zeta})_{,\bar\zeta}+\gamma^2\lambda \omega_{,\bar\zeta}\ .\label{27u}
\ee
Given  functions $\alpha$,  $\lambda$, $\gamma$, $\omega$, $T^{\ 3}_{3}$ and $T_{\zeta}^{\ \zeta}$ equation (\ref{27u}) defines $\alpha T_{\bar\zeta\bar\zeta}$ as an integral.
\end{itemize}
\end{prop}
\begin{proof}
Metric (\ref{27}) is the most general metric with a single continuous symmetry $\p_{\varphi}$. If data are independent of $\varphi$ then in the basis $\theta^a=dx^a$, $\theta^3=d\varphi+\beta_a dx^a$ 
the  momentum constraint $\nabla_jT_3^{\ j}=0$ takes very simple form
\begin{equation}
(\alpha\sqrt{|\tilde g|} T^{\ a}_3)_{,a}=0\ ,\label{27a}
\end{equation}
where $|\tilde g|=\det{g_{ab}}$. 
General solution of (\ref{27a}) is given by (\ref{28}). Substituting it into $\nabla_jT_a^{\ j}=0$ yields equation (\ref{27b}).

In dimension 2 one can always find coordinates $x^a=x,y$ such that
\be
g_{ab}=\gamma^2\delta_{ab}\ .\label{28z}
\ee
In this case a complex combination of equations (\ref{27b})  written in complex coordinates $\zeta=x+iy$, $\bar\zeta=x-iy$ yields equation (\ref{27u}). The assumption of analyticity allows to integrate the r. h. s. of (\ref{27u}) over $\zeta$. Note that $T_{\bar\zeta\bar\zeta}$ is a complex function and $T_{\zeta}^{\ \zeta}$ is real. 
For non-analytic fields one can use the formula
\be
2\alpha T_{\bar\zeta\bar\zeta}=-\frac{1}{4\pi}\int{d^2x'\frac{F(x'^a)}{(\bar\xi'-\bar\xi)}}+h(\bar\xi)
\ee
where $F$ denotes the rhs of (\ref{27u}).

\end{proof}
An advantage of formula (\ref{27u}) is that it treats all cases on equal footing. However, in order to obtain a more explicit description of solutions it is  useful to consider separately the cases $\alpha=const$ and $\alpha\neq const$.
\begin{prop}\label{p1a}
 If $\alpha\neq const$ then 
\begin{itemize}
\item
Tensor $T^{ab}$ uniquely corresponds  to functions $T^a,\ T^0$ such that
\be\label{a1}
T^{ab}=\eta^a T^b+\zeta^a\eta^b T^c\xi_c+\xi^a\xi^bT^0
\ee
\be\label{a2}
T^a= T^{ab}\eta_b\ ,\ \ \ T^0=\xi_{a}T^{ab}\xi_{b}\ ,
\ee
where 
\be\label{a3}
\xi_b=\frac{\alpha_{,b}}{\triangledown\alpha}\ ,\ \ \eta^a=\eta^{ab}\xi_b\ ,\ \ \triangledown\alpha=(\alpha^{,c}\alpha_{,c})^{\frac 12}\ .
\ee
\item 
If functions $\alpha$ and $\triangledown\alpha$ are independent then equation (\ref{27b}) determines $T^0$ and $T^{\ 3}_{3}$ in terms of other data  
\be\label{a5}
\alpha\kappa T^0=(\alpha T^a)_{|a}+\alpha\xi_{b}T^b\xi^a_{\ |a}-\lambda \eta^a\omega_{,a}
\ee
\be\label{a4}
(\triangledown\alpha) T^{\ 3}_{3}=(\alpha T^0 \xi^b+\alpha \xi_aT^a\eta^b)_{|b}-\alpha\eta^b\xi_{b|a}T^a -\lambda \xi^a\omega_{,a}\ ,
\ee
where
\be
\kappa=\eta^{bc}\alpha_{,b}(\triangledown \alpha)_{,c}(\triangledown\alpha)^{-2}\ .
\ee
\item
If functions $\alpha$ and $\nabla\alpha$  are dependent then equation (\ref{27b}) is equivalent to (\ref{a4}) and the equation
\be\label{a8}
(\alpha \eta_bT^b\eta^a)_{|a}=-(\alpha \xi_bT^b\xi^a)_{|a}-\alpha \xi_{b}T^b \xi^a_{\ |a}+\lambda \eta^a\omega_{,a}
\ee
which yields an integral expression for one of the functions   $\eta_bT^b$, $\xi_bT^b$ or $\omega$.
\end{itemize}
\end{prop}
\begin{proof}
Vectors $\xi^a$ and $\eta^a$ form an orthonormal basis on surfaces $\varphi=const$. If  $T^a$ and $T^0$ are defined by (\ref{a2}) then (\ref{a1})  can be easily proved by taking contractions of $T^{ab}$ with  $\eta_a$ and  $\xi_a$.
 Similar contractions of equation (\ref{27b}) lead, respectively, to
\begin{equation}
\eta^{a}(\alpha T^{\ b}_{a})_{|b}=\lambda \eta^a\omega_{,a} \label{a7}
\end{equation}
and
\begin{equation}
\xi^a(\alpha T^{\ b}_{a})_{|b}= (\nabla\alpha)T^{\ 3}_{3}+\lambda \xi^a\omega_{,a}\ .\label{a6}
\end{equation}
Substituting (\ref{a1}) into (\ref{a7}) and (\ref{a6}) yields (\ref{a5}) and (\ref{a4}). Since $\alpha^{,c}\alpha_{,c}\neq 0$ equation (\ref{a4}) can be treated as a definition of $T^{\ 3}_{3}$. If functions $\alpha$ and $\triangledown\alpha$ are independent then $\kappa\neq 0$ and (\ref{a5}) defines $T^0$. If they are dependent then $\kappa= 0$ and equation (\ref{a5}),  with $T^a$ decomposed in the basis $(\xi^a,\eta^a)$, reduces  to equation (\ref{a8}). In order to show that the latter equation can be easily integrated let us choose coordinates $x^a=x,y$ such that $\alpha=y$ and
\be
g=\gamma^2dx^2+\sigma^2dy^2+y^2(d\varphi+\beta_adx^a)^2\ .\label{a13}
\ee
Then $\xi_aT^a=\sigma\gamma T^{xy}$, $\eta_aT^a=\gamma^2 T^{xx}$ and  equation (\ref{a8}) reads
\be\label{a10}
y\big(\sigma T^{\ x}_x\big)_{,x}=-\gamma^{-1}\big(y\sigma\gamma T^{\ y}_x\big)_{,y}+\sigma\lambda\omega_{,x}\ .
\ee
Equation (\ref{a10}) is equivalent to an integral expression for $T^{xx}$, $T^{xy}$ or $\omega$. Alternatively one can derive from it an explicit expression for $\sigma$ or $\gamma$ in terms of other variables.

\end{proof}
The most general symmetric data satisfying the momentum constraint  are given by (\ref{27}), (\ref{28}), (\ref{a5}) and (\ref{a4}). They depend on 9 functions of 2 coordinates. Three of them  can be fixed by means of a transformation of coordinates $x^a$ and a shift of $\varphi$. Thus, there are 6 functions which cannot be gauged away. Their number decreases to 5 if the Hamiltonian constraint is imposed.

Now, let us assume that 
 $\alpha$=const.  Without a loss of generality   one can set $\alpha=1$.  Function $T^{\ 3}_{3}$ is arbitrary since now it is not  present in (\ref{27b}). Since $\alpha_{,a}=0$ there is no hint how to choose coordinates $x$, $y$ in order to represent solutions of equation (\ref{27b}) in a simple way.   The following proposition  shows how it can be done. We omit the proof which is straightforward.
\begin{prop}\label{p3}
 If $\alpha=1$ then in coordinates such that 
\be
g=dx^2+\sigma^2dy^2+(d\varphi+\beta_adx^a)^2\label{a15}
\ee
equation (\ref{27b}) is equivalent to
\be
\sigma T^{\ y}_{y\ ,y}=-(\sigma^3T^{xy})_{,x}+\lambda\sigma \omega_{,y}\label{a11}
\ee
and
\be
(\sigma T^{\ x}_x)_{,x}=\sigma_{,x}T^{\ y}_y-(\sigma T^{xy})_{,y}+\lambda\sigma\omega_{,x}\ .\label{a12}
\ee
Equation (\ref{a11}) defines $T^{\ y}_y$  and, consecutively, equation (\ref{a12}) defines $T^{\ x}_x$.
\end{prop}

Formulas (\ref{28}), (\ref{a5}) and (\ref{a4}) allow to obtain generic  solutions of the momentum constraint in an explicit way. Irrespectively of that one can  find particularly simple solutions of (\ref{27b}) making an ansatz. For instance,
 let  $T^{33}$=0, $\alpha T^a_{\ a}$=const and
\begin{equation}
g=\sigma^2(dx^2+dy^2)+\alpha^2 d\varphi^2\ .\label{a27}
\end{equation}
Then equation (\ref{27b}) reads
\begin{equation}
(\alpha\sigma^4 \tilde T^{ab})_{,b}=0\ ,\label{28f}
\end{equation}
where $\tilde T^{\ b}_{a}$ is the traceless part of $T^{\ b}_{a}$. It follows from (\ref{28f}) that there is a function $f$ such that 
\begin{equation}
\alpha T_{a b}=f_{,ab}+cg_{ab}\ ,\ \ f_{,xx}+f_{,yy}=0\ ,\ \ c=const\ .\label{28g}
\end{equation}
One can slightly generalize these solutions assuming 
\begin{equation}
\alpha T_{a b}=f_{,ab}+h(\sigma)\delta_{ab}\ ,\label{28j}
\end{equation}
\be
f_{,xx}+f_{,yy}=2\sigma^3 h_{,\sigma}\ ,\label{28n}
\ee
where $h$ is a function of $\sigma$.  These data can be also obtained from (\ref{27u}) under the assumption that $\alpha T_{\zeta}^{\ \zeta}$ is a function of $\sigma$.
Note that for a nontrivial dependence of the  r. h. s. of (\ref{28n}) on $\sigma$ this equation can be considered as  a definition of $\sigma$ in terms of the function $f$.

Other simple solutions of (\ref{27b}) with $\alpha\neq$const can be obtained if $ T^{33}=0$ and $\pm \alpha T^{ab}$ has the form of   the energy-momentum tensor of a scalar field $f$ with a potential $V(f)$
\begin{equation}
\pm \alpha T_{a b}=f_{,a}f_{,b}-(\frac 12f^{,c}f_{,c}+V)g_{ab}\ .\label{28m}
\end{equation}
In this case equation (\ref{27b}) is  satisfied provided  that the scalar field equation $f^{|a}_{\ \ |a}=V_{,f}$ is fulfilled. For metric $g_{ab}$ of the form (\ref{a27}) this equation reads
\be
f_{,aa}=\sigma^2V_{,f}\ .\label{28u}
\ee
If $V_{,f}\neq$0 equation (\ref{28u})  defines $\sigma$ in terms of $f$. If $V_{,f}=$0 then $f=Re\ h(\zeta)$, where $h$ is a holomorphic function of $\zeta$=$x+iy$. In both cases the data  can be also obtained from (\ref{27u}) under the assumption that $\alpha T_{\zeta}^{\ \zeta}$ is a function of $f$ and $\alpha T_{\bar\zeta\bar\zeta}=\pm (f_{,\bar\zeta})^2$.

 Formulas (\ref{28j})-(\ref{28n}) or (\ref{28m})-(\ref{28u})  define also particular  solutions of the momentum constraint if  $\alpha$=const. In this case the function $T^{33}$ can be arbitrary.

\begin{cor}
Simple solutions of equation (\ref{27b}) are given by  (\ref{a27}) and (\ref{28j})-(\ref{28n}) or (\ref{28m})-(\ref{28u}). If $\alpha\neq$const then $T^{33}=0$.
\end{cor}

In order to solve the Hamiltonian constraint by means of the conformal method it is important to have asymptotically flat solutions of the momentum constraint with $H=0$. Then, by means of a conformal transformation one can obtain either $\alpha=1$ or ($\alpha\neq const$, $\kappa\neq 0$) or ($\alpha\neq const$, $\kappa=0$). In all these cases condition $H=0$ imposes an extra constraint on the free functions. The simplest situation is for $\alpha=1$ since then one can adjust 
$T^{\ 3}_3$ to get $T^i_{\ i}=0$. In order to obtain an asymptotically flat metric  another conformal transformation is necessary. For instance, one can multiply metric (\ref{a15}) by $e^{2x}$ and interpret $e^x$ as a distance from the symmetry axis. The new metric is asymptotically flat if $\sigma e^x\rightarrow 1$ for $e^x\rightarrow\infty$ or $y\rightarrow\infty$. Vanishing of the new  tensor $T'^{ij}$ at infinity can be assured by an appropriate condition on $T^{xy}$ and $\omega$.

If $\alpha\neq const$ then condition $H=0$ is equivalent to 
\be
T^{\ 3}_3 +T^0+\eta_aT^a=0\ .\label{a14}
\ee
 If $\alpha$ and $\triangledown\alpha$ are independent then this equation becomes a second order linear equation for one of the functions $T^a$. This is more complicated situation than in the case $\alpha=1$, however, now one can expect that the  data are asymptotically flat without a necessity of a conformal transformation.
If $\alpha\neq const$ and $\kappa=0$ then condition (\ref{a14}) can be treated as an equation for $T^0$. If metric is in the form (\ref{a13}) then this equation leads to an expression for $(y^2\sigma^{-1}\gamma T^0)_{,y}$  hence $T^0$ can be determined. Metric (\ref{a13}) may be asymptotically flat without performing a conformal transformation. For instance, it is sufficient that $\sigma\rightarrow 1$ and $\gamma\rightarrow 1$ if $x^2+y^2\rightarrow\infty$.

Equation (\ref{27b}) is trivially satisfied and $H=0$ if
\be\label{45}
T^{ab}=T^{33}=0\ ,\ \ \lambda=0\ .
\ee
In this case the symmetry vector $\p_{\varphi}$ is nontwisting and we can change the coordinate $\varphi$ to obtain $\beta_a=0$ in  metric (\ref{27}).
Condition (\ref{45})  is realized on  constant time surfaces in the Kerr metric and in other stationary axially symmetric metrics defined by solutions of the Ernst equation \cite{k}. Indeed, these metrics read 
\begin{equation}
g^{(4)}=g_{AB}dx^Adx^B+g_{ab}dx^adx^b,\label{26}
\end{equation}
where  $x^A=t,\varphi$ and the  metric components  depend  only on  $x^a=r,\theta$. On the surface $t=$const metric (\ref{26}) takes the form 
(\ref{27}).
Since  $k=k^A\partial_A$ formula (\ref{-2}) yields  the exterior curvature tensor of the form  $u_adx^a d\varphi$, where  $u_a$ are functions of $x^a$. Thus, $K^{ab}=K^{33}=0$, hence also $H=0$ and condition (\ref{45}) follows. If the vacuum Einstein equations are satisfied components $K^{a3}$ must have  form (\ref{28}). For instance, 
in the case of the Kerr metric one obtains 
\be 
g=\rho^2\Delta^{-1}dr^2+\rho^2d\theta^2+\rho^{-2}\Sigma^2\sin^2{\theta}d\varphi^2\label{8a}
\ee
and (\ref{28}) with
\be
\omega=4aM\rho^{-2}[2(r^2+a^2)+(r^2-a^2)\sin^2{\theta}]\cos{\theta}\ ,\label{10a}
\ee
where
\be
\rho^2=r^2+a^2\cos^2{\theta}\ ,\ \ \Delta=r^2-2Mr+a^2\ ,\ \Sigma^2=(r^2+a^2)^2-a^2\Delta\sin^2{\theta}\ .\label{9a}
\ee
Metric (\ref{8a}) tends to the flat metric if $r\rightarrow \infty$ and function $\omega$  satisfies
\be
\omega\rightarrow 4aM(2+\sin^2{\theta})\cos{\theta}\ \ \ \texttt{if}\ \ \ r\rightarrow\infty\ .\label{h11}
\ee
Condition (\ref{h11}) assures fast vanishing of tensor $T^{ij}$ when $r$ increases. It can be used to define the Kerr-like asymptotical behaviour of the initial data.

If $\p_{\varphi}$ is interpreted as the axial symmetry then initial  data should be regular on the symmetry axis. Hence, in spherical like coordinates $r$, $\theta$, $\varphi$ metric should behave like
\be
g\backsimeq g_{rr} dr^2+g_{\theta\theta} (d\theta^2+\sin^2{\theta}d\varphi^2) \ \ \texttt{if}\ \sin{\theta}\rightarrow 0\ ,\label{h4}
\ee
where $g_{rr}$ and $g_{\theta\theta}$  are positive  functions. 
A more general expression is allowed in the case of the exterior curvature. In particular,  $\omega$ should satisfy
\be
 \omega_{,r} \sim \sin^4{\theta}\ ,\ \ \ \ \omega_{,\theta}\sim \sin^3{\theta}\ \ \ \texttt{if}\ \sin{\theta}\rightarrow 0\ .\label{h6}
\ee
Obviously, conditions (\ref{h4}) and (\ref{h6}) are satisfied by the Kerr data.

In general, symmetric data undergoing condition (\ref{45})  are not related to stationary axially symmetric metrics of the form (\ref{26}). Note that these metrics are fully determined by the complex Ernst potential which has to satisfy the Ernst equation. Hence, the corresponding initial data on $t=$const is also determined by  this potential. Generic axially symmetric data satisfying  (\ref{45})  do not share this property. 

\section{Nonsymmetric solutions of the momentum constraint}\label{ssec:uv_ind}
Tensor  $T^{ij}$ can be written in the form
\begin{equation}\label{1}
T^{ij}=\rho u^{(i}v^{j)}+pg^{ij},
\end{equation}
where vectors $u=u^i\p_i$ and $v=v^i\p_i$ are real  or  complex conjugated, $u=\bar v$.
To this end we  define $p$ as a real  solution of the equation
\begin{equation}
\det{(T^{ij}-pg^{ij})}=0.\label{2}
\end{equation}
Since (\ref{2}) is a third order polynomial equation for $p$  it admits one or three real solutions. If $p$ is one of them tensor $T^{ij}-pg^{ij}$ has the signature $0,\pm 1,\pm 1$. Hence, decomposition (\ref{1}) follows (not uniquely if there are 3 different real solutions of (\ref{2})). Vectors $u$ and $v$ are orthogonal to the eigenvector $e=e^i\p_i$ corresponding to $p$. One can rescale them in order to obtain $\rho=\pm 1$  but we keep $\rho$ in general form for a later convenience.

 The momentum constraint can  be explicitly solved under particular assumptions on decomposition (\ref{1}) . In this section we assume that the eigenvector  $e$ is twist free. In this case  there exists a function $f$ such that $e_i\sim f_{,i}$ and $u$ and $v$ are tangent to surfaces $f=$const. If they are real we can choose coordinates $x^i=x,y,z$  such that $u\sim \p_x$, $f=y$ and vector $v$ is spanned by $\p_x$ and $\p_z$. If $u$ and $v$ are linearly independent we can still change coordinate $x$  to obtain $v\sim \p_z$. Hence, without loss of generality, we can assume 
\begin{equation}
u=\p_x,\quad  v=\p_z\label{19}
\end{equation}
or
\begin{equation}
v=u= \p_z\ .\label{19a}
\end{equation}
If vectors $u$ and $v$ are complex we obtain (\ref{19}) with complex coordinates $x,z$ such that $x=\bar z$. 

Under assumption (\ref{19}) or (\ref{19a})  solutions of the momentum constraint can be expressed in terms of free functions and integrals over them.  In this section we will present only a family of  solutions  which can be written in a relatively simple way. In general they have no symmetry, however, they contain a subclass of data obtained in Section 2. 

In order to describe solutions  satisfying (\ref{19}) let us first introduce some notation.  
\begin{defi}
Given functions  $\beta$ and $\gamma$ such that matrix $\gamma_{,AB}$ is nondegenerate we define $\gamma_{AB}$,
 $|\tilde \gamma|$ and $\beta^A$  by
\be
\gamma_{AB}=\gamma_{,CD}\sigma^C_{\ A}\sigma^D_{\ B}\ ,\ \ |\tilde \gamma|=(\det{\gamma_{,AB}})^{\frac 12}\ ,\ \ 
\gamma_{AC}\beta^{C}=\beta_{,C}\sigma^C_{\ A}\ ,\label{20ab}
\ee
where $x^A=x,z$ and $\sigma^A_{\ B}$=diag$(1,-1)$.
\end{defi}
This notation helps to describe the following  class of solutions of the momentum constraint.
\begin{prop}\label{pr1}
The momentum constraint is satisfied  by
\begin{equation}
g=(\alpha\rho|\tilde \gamma|)^{-\frac 25}[\alpha^2 dy^2+\gamma_{AB}(dx^A+\beta^Ady)(dx^B+\beta^Bdy)]\label{20d}
\end{equation}
\be
T^{ij}=\rho \delta^{(i}_x\delta^{j)}_z+pg^{ij}
\ee
provided that  $\gamma_{,AB}$ has the Euclidean signature and functions $\gamma$, $\beta$, $p$, $\alpha>0$, $\rho>0$   satisfy either
\be
p=p(g_{xz})\ ,\ \ \rho=2\frac{d p}{dg_{xz}}\label{20e}
\ee
or
\be
p=const\ ,\ \ \gamma_{,xz}=c (\alpha\rho|\tilde \gamma|)^{\frac 25}\ ,\ \ c=const\ .\label{20f}
\ee
All coordinates and functions  are either real or  $x=\bar z$ and $\beta$ is purely imaginary while other functions are real. Three  of the functions are arbitrary up to positivity conditions.
\end{prop}
\begin{proof}
In  the case (\ref{19})  the momentum constraint reads
\begin{equation}
(\rho|g|g_{ix})_{,z}+(\rho|g|g_{iz})_{,x}=|g|(\rho g_{xz,i}-2p_{,i})\ .\label{20}
\end{equation}
  In order to simplify (\ref{20}) we  assume  that the r. h. s.   of (\ref{20}) vanishes
\be
2p_{,i}=\rho g_{xz,i}\ .\label{20a}
\ee
In this case it follows from (\ref{20})  that there exist functions $\gamma_i$ such that 
\be\label{20k}
\rho|g|g_{ix}=\gamma_{i,x}\ ,\ \ \rho|g|g_{iz}=-\gamma_{i,z}\ .
\ee
These equations are compatible iff $\gamma_{z,x}=-\gamma_{x,z}$, hence there exists a function $\gamma$ such that $\gamma_x=\gamma_{,x}$ and $\gamma_z=-\gamma_{,z}$. Substituting these expressions into (\ref{20k}) and denoting $\gamma_y$ by $\beta$ yields
\begin{equation}
\rho|g|g=\alpha^2 dy^2+\gamma_{AB}(dx^A+\beta^Ady)(dx^B+\beta^Bdy)\ ,\label{20j}
\end{equation}
where $\gamma_{AB}$ and $\beta^A$ are given by Definition 2.2.
Function $\alpha$ can be   arbitrary positive since $g_{yy}$ is not present in (\ref{20}). From (\ref{20j}) one can  calculate $\det{g_{ij}}$ and  obtain metric in the form (\ref{20d}). Still equation (\ref{20a}) has to be satisfied.
If $g_{xz}\neq$const it  yields (\ref{20e}). If $g_{xz}$=const  one obtains (\ref{20f}).

In the case (\ref{20e}) functions $\alpha$, $\beta$ and $\gamma$ are arbitrary up to  $\alpha>0$ and positivity of  $\gamma_{,AB}$. It is also true 
 in the case  (\ref{20f}) with $c\neq$0 since then  $\rho$ can be expressed in terms of other functions. In the case (\ref{20f}) with $c$=0  function $\gamma$ splits into the sum of functions of 2 variables  but then $\alpha$, $\beta$ and $\rho$ are arbitrary. Thus, in each case there are 3  functions of 3 coordinates which are arbitrary up to conditions of positivity.

\end{proof}
 In order to prepare ground for studies of the Hamiltonian constraint  let us identify solutions  from Proposition \ref{pr1} with $H=0$. If $H=0$ and $\gamma_{,xz}\neq 0$ we can apply a conformal transformation to obtain 
\begin{equation}
|\tilde\gamma|^2\alpha^2=|\gamma_{,xz}|^5\label{22s}
\end{equation}
\begin{equation}
\rho=1\ ,\ \  p=\frac 13\operatorname{sgn}(\gamma_{,xz})\ .\label{22f}
\end{equation} 
Metric and $T^{ij}$  now read
\begin{equation}
g=\frac{(\gamma_{,xz})^4}{|\tilde\gamma|^2} dy^2+|\gamma_{,xz}|^{-1}\gamma_{AB}(dx^A+\beta^Ady)(dx^B+\beta^Bdy))\label{21q}
\end{equation}
\be
T^{ij}=\delta^{(i}_x\delta^{j)}_z+\frac 13\operatorname{sgn}(\gamma_{,xz})g^{ij}\ .\label{21y}
\ee

If $\gamma_{,xz}=p=0$ then by means of a conformal transformation one  obtains 
\be
T^{ij}=(\alpha fh)^{-1}\delta^{(i}_x\delta^{j)}_z\label{23c}
\ee
and 
\begin{equation}
g=\alpha^2 dy^2+ f^2(x,y)(dx+\beta_{,x}dy)^2+ h^2(y,z)(dz-\beta_{,z}dy)^2\ .
\label{25}
\end{equation}
Here $f$ and $h$ are functions of 2 coordinates as indicated.
A change of coordinates $x$ and $z$ allows to transform $f$ and $h$ to any  pair of nonzero functions of this type, e.g. to constant functions. However, it is convenient to keep them arbitrary until conditions of asymptotical flatness are imposed.
\begin{cor}\label{cor1}
 The momentum constraint with $H=0$ is satisfied by data  given by  (\ref{21q})-(\ref{21y}) or  (\ref{23c})-(\ref{25}), where $\gamma_{AB}$ and $\beta^A$ are related to functions $\beta$ and $\gamma$ according to Definition 2.2. Functions $f$ and $h$ can be transformed to $f=h=1$.
\end{cor}
 Data (\ref{23c})-(\ref{25}) include (up to conformal transformations)  a subclass of  invariant data  satisfying  condition (\ref{45}). For them $z=\varphi$ and $T^{ij}\sim \delta^{(i}_x\delta^{j)}_z$  in coordinates $x^i$ such that $\omega=F(y)$. Metric (\ref{27}) is conformally equivalent to metric of the form (\ref{25}) with $\alpha_{,z}=\beta_{,z}=h_{,z}=0$.

 Thanks to this property one can deduce conditions which assure asymptotical flatness of data (\ref{23c})-(\ref{25})
\be
f^2=\frac{\sin{y}}{x^2}\ ,\ \ h^2=\sin^3{y}\label{25ka}
\ee
\be\label{25k}
\alpha^2\rightarrow \sin{y}\ ,\ \ \beta_{,x}\rightarrow 0\ ,\ \ \beta_{,z}\rightarrow 0\ \ \ if\ \ \ x\rightarrow\infty\ .
\ee
Under condition  (\ref{25ka}) the conformal transformation (\ref{eqn:conf}) with $\psi^2=x\alpha^{-1}$ yields
\be
T^{ij}=\frac{\tilde\alpha^2}{x^4}\delta^{(i}_x\delta^{j)}_z\label{25b}
\ee
and
\be
g=x^2 dy^2+ \tilde\alpha(dx+\beta_{,x}dy)^2+ x^2\tilde\alpha\sin^2{y}(dz-\beta_{,z}dy)^2\ ,
\label{25g}
\ee
where $\tilde\alpha=\alpha^2/\sin{y}$. Conditions (\ref{25k}) can be replaced by
\be\label{25l}
\tilde\alpha\rightarrow 1\ ,\ \ \beta\rightarrow 0\ \ \ if\ \ \ x\rightarrow\infty\ .
\ee
It is clear from the above assumptions that, asymptotically, $x,y,z$ become  spherical coordinates $r,\theta,\varphi$ of the flat metric.
Data given by (\ref{25b}) and (\ref{25g}) generalize the Kerr data. They  contain 2 free functions ($\tilde\alpha$ and $\beta$) of 3 coordinates. In general, these data have no symmetries. They are asymptotically flat under conditions (\ref{25l}) provided that derivatives of $\tilde\alpha$ and $\beta$ vanish  sufficiently fast.

Now, we turn to initial data obeying condition (\ref{19a}).  
\begin{prop}\label{pr2}
 The momentum constraint is satisfied by the data given by
\be
\pm T^{ij}= \delta^{i}_z\delta^{j}_z+(c+\frac{1}{2}\beta |\tilde g|^{-\frac 23})g^{ij}\label{3s}
\ee
\begin{equation}
g=\beta |\tilde g|^{-\frac 23}(dz+\beta_adx^a)^2+g_{ab}dx^adx^b\label{3c}
\end{equation}
or
\be
\pm T^{ij}=\beta |\tilde g|^{-1}\delta^{i}_z\delta^{j}_z+cg^{ij}\label{3u}
\ee
\begin{equation}
g=(dz+\beta_adx^a)^2+g_{ab}dx^adx^b\ ,\label{3b}
\end{equation}
where $\beta_{a,z}=\beta_{,z}=0$, $\beta>0$, $c=const$, metric $g_{ab}$ has the Euclidean signature and $|\tilde g|$ denotes $\det{g_{ab}}$. Function $\beta$ can be transformed to 1.
\end{prop}
\begin{proof}
In  case  (\ref{19a}) the momentum constraint reads
\begin{equation}
(\rho|g|g_{iz})_{,z}=\frac 12\rho|g|g_{zz,i}-|g|p_{,i}\label{3}
\end{equation}
where $|g|=(\det{g_{ij}})^{1/2}$. It can be solved in full generality  in terms of integrals. In order to obtain simpler solutions we first 
assume  
\be
p=const+\frac 12\rho g_{zz}\ ,\ \ \rho=\pm 1\ .\label{21x}
\ee
In this case it follows from (\ref{3}) that 
\be
g=|g|^{-1}\beta^3(dz+\beta_adx^a)^2 +g_{ab}dx^adx^b\ ,\ \ \beta_{a,z}=\beta_{,z}=0\ .\label{21t}
\ee
Taking  determinant of this metric  yields 
\be\label{21h}
|g|=\beta^2|\tilde g|^{-\frac 23}\ .
\ee
From (\ref{21t}) and (\ref{21h}) one obtains metric in the form (\ref{3c}). Equation (\ref{3s}) follows from (\ref{21x}) and (\ref{3c}).

If we assume $p=const$  it is  convenient to work with the unit vector $v=\p_z$. Then $g_{zz}=1$ and from (\ref{3})
one obtains   $g_{az,z}=0$ and $\rho\sqrt{|g|}=\beta(x^a)$.  Hence, expressions (\ref{3u}) and (\ref{3b}) follow.

In  solutions described by this proposition    components of the metric $g_{ab}$ cannot be in general gauged away since coordinate transformations are strongly restricted by the  form of $T^{ij}$.  One can accommodate $\beta$ in $\det{g_{ab}}$ by a change of coordinates $x^a$. We keep arbitrary $\beta$ since then the asymptotical flatness conditions    are simpler.

\end{proof}
Data with $H=0$ are present in both classes of solutions presented in Proposition \ref{pr2}. They are jointly given by
\begin{cor}\label{cor2}
The data
 \be
T^{ij}= c(\delta^{i}_z\delta^{j}_z-\frac 13 g^{ij})
\ee
\begin{equation}
g=(dz+\beta_adx^a)^2+g_{ab}dx^adx^b\ ,
\end{equation}
where 
$c=const$ and $\beta_{a,z}=(\det{g_{ab}})_{,z}=0$, satisfy the momentum constraint and $H=0$.
\end{cor}
Above solutions can be asymptotically flat. For instance, let $z=\ln{r}$, $\beta_a=0$  and $x^a=\theta,\varphi$ be spherical angles. If
\be
\det{g_{ab}}=\sin^2{\theta}
\ee
and 
\begin{equation}
g_{ab}dx^adx^b\rightarrow d\theta^2+\sin^2{\theta}d\varphi^2\ \ if\ \ r\rightarrow \infty
\end{equation}
then  conformal transformation (\ref{eqn:conf}) with $\psi^2=r$ yields asymptotically flat metric. 

\section{Hamiltonian constraint}\label{ham}
In this section and in the next one we will apply the Lichnerowicz-Choquet-Bruhat-York conformal method to solutions of the momentum constraint from Sections 2 and 3. Before we do it let us shortly discuss how to exploit free functions in these solutions in order to solve the Hamiltonian constraint if $H\neq 0$. For instance, this can be easily done for solutions with $g_{\varphi\varphi}=1$  described by Proposition \ref{p3}. Then the function $T^{\ \varphi}_{\varphi}$ is not involved in the momentum constraint and it can be defined by the following equation equivalent to (\ref{eqn:scalar})
 \be\label{90}
(T^{\ \varphi}_{\varphi}-T^{\ a}_{a})^2=2R+3(T^{\ a}_{a})^2-2T_{ab}T^{ab}-4T_{a\varphi}T^{a\varphi}\ ,
\ee
provided that the r. h. s. of (\ref{90}) is nonnegative. Unfortunately, since $H\neq 0$ and conformal transformations are not allowed, initial metric cannot be asymptotically flat in this case. Another simple example concerns solutions described in the last point of Proposition \ref{p1a}. Now,  the Hamiltonian constraint can be considered as an ordinary differential equation for  the function $T^0$.  In coordinates (\ref{a13}) it reads
\be
\gamma^{-1}(\gamma T^0)_{,y}=F\pm 2 y^{-1}(G+T^0\eta_aT^a)^{\frac 12}\ ,\label{91}
\ee
where $F$ and $G$ are expressions independent of $T^0$
\ba
F&=&(\gamma^{-1}\gamma_{,y}+y^{-1})\eta_aT^a-(\beta\gamma)^{-1}(\beta^2l_aT^a)_{,x}\\
G&=&\frac 12R-(l_aT^a)^2-T_{a\varphi}T^{a\varphi}\ .
\ea
 Unfortunately, we are not able to prove existence of global solutions of   equation (\ref{91}).  If they exist one can obtain asymptotically flat data in this case ($\beta$ and $\gamma$ should tend to 1 at infinity).

From now on we will assume that  $H=0$ and we focus on asymptotically flat data. In this case
  the conformal method is effective if one can show that the Lichnerowicz  equation (\ref{eqn:scalar1}) has a solution  $\psi$ which is  positive everywhere   and tends to 1  at infinity.   Our considerations are mainly based on the existence theorems of Maxwell \cite{ma} (note that the exterior curvature in \cite{ma} has the opposite sign with respect to that defined by (\ref{-2})).They are applicable if components of seed data and the exterior curvature belong to  appropriate weighted Sobolev spaces $W_{\delta}^{k,p}$. Metric doesn't have to be conformally flat.   Existence and uniqueness of $\psi$ depends crucially on the positivity of the   Yamabe type invariant 
\be\label{h15}
\lambda_g:=\operatorname{inf}_{f\in C^\infty_c} \frac{\int_S(8|\nabla f|^2+R^{}f^2)dv_g}{\|f\|^2_{L^6}}>0\ ,
\ee
where $dv_g$ (also present in the norm of $f$) is the volume element  corresponding to seed metric $g$.

In order to formulate the results of Maxwell  we define asymptotically flat initial data  in the following way.
\begin{defi}
Let an initial surface $S$ be a union of a compact set and so-called asymptotically flat ends, which are  diffeomorphic to a completion $E$ of a ball  in $R^3$. We say that data ($g,K$) are asymptotically flat of class $W_{\delta}^{k,p}$ if  the following conditions are satisfied 
\be
g\in W_{loc}^{k,p}(S)\ ,\ \ (g_{ij}-\delta_{ij})\in W_{\delta}^{k,p}(E)\ ,\ \ K_{ij}\in W_{\delta-1}^{k-1,p}(S)\label{h}
\ee
where $k\geq 2$, $kp>3$, $\delta<0$ and indices $i,j$ correspond to   the Cartesian coordinates of $E$. 
\end{defi}
For instance, conditions (\ref{h}) are satisfied, with $\delta>-\epsilon$ and arbitrary $p$, 
if $g$ and $K$ are  fields of  class $C^k$ and $C^{k-1}$, respectively,  and they satisfy the following conditions which are often used in relativity
\be
g_{ij}=\delta_{ij}+0_k(r^{-\epsilon})\ ,\ \ K_{ij}=0_{k-1}(r^{-1-\epsilon})\ .\label{5}
\ee
Here $\epsilon$ is a positive constant, $r$ is the radial distance in $R^3$  and we write $f=0_k(r^{-\epsilon})$ if derivatives of $f$  of order $n\leq k$ fall off as $r^{-n-\epsilon}$ when $r\rightarrow \infty$.
Note that conditions (\ref{h}) with $\delta<-1/2$ or (\ref{5}) with $\epsilon>1/2$ are sufficient to define the ADM energy-momentum \cite{b}.

The existence theorem of Maxwell can be formulated in the following form 
\begin{thm}[Maxwell]\label{ma}
 Let ($S,g$) be a complete Riemannian manifold without boundary and let $(g,K)$  be a traceless ($H$=0) solution of the momentum constraint which is asymptotically flat  of class $W_{\delta}^{k,p}$. Then there exists a solution of the Lichnerowicz equation (\ref{eqn:scalar1}) such that $\psi>0$ and  $(\psi-1)\in W_{\delta}^{k,p}$ if and only if  $\lambda_g>0$. If it exists it is unique and the conformally transformed data (\ref{eqn:conf}) satisfy all  constraint equations and are asymptotically flat of class  $W_{\delta}^{k,p}$.
 \end{thm}
A direct verification of condition (\ref{h15}) is rather difficult. We will show that it can  be replaced by a simpler condition if the Euclidean type Sobolev inequality is satisfied on $S$. Let us  decompose  the Ricci scalar of $g$  into a positive and negative part
\be\label{h16a}
 R=R_++R_-\ ,
\ee
 where $R_-=0$ at points where  $R\geq 0$ and $R_+=0$  if $R\leq 0$.
\begin{prop}\label{prop0}
If   the Sobolev inequality 
\be\label{h16}
\|f\|_{L^6}\leq A\|\nabla f\|_{L^2}\ ,\ \ A=const>0
\ee
is satisfied and
\be\label{h18}
\|R_-\|_{L^{3/2}}<\frac{8}{A^2}
\ee
then $\lambda_g>0$.
\end{prop}
\begin{proof}
 From the H$\ddot{\texttt{o}}$lder inequality one obtains
\be\label{h17}
\int_S{|R_-|f^2dv_g}\leq \|R_-\|_{L^{3/2}}\|f\|_{L^6}^2\ .
\ee
Since  $R\geq R_-$ inequality (\ref{h17}) leads to
\be\label{h17b}
\int_S{R f^2dv_g}\geq -\|R_-\|_{L^{3/2}}\|f\|_{L^6}^2\ .
\ee
From (\ref{h17b}) and (\ref{h16}) it follows that 
\be
\int_S(8|\nabla f|^2+R^{}f^2)dv_g\geq (\frac{8}{A^2}-\|R_-\|_{L^{3/2}})\|f\|_{L^6}^2\ .
\ee
Hence, 
\be
\lambda_g\geq (\frac{8}{A^2}-\|R_-\|_{L^{3/2}})
\ee
 and  $\lambda_g$ is positive if (\ref{h18}) is satisfied.

\end{proof}
In the 3-dimensional flat Euclidean  space the Sobolev inequality (\ref{h16}) is satisfied and $R=0$, hence condition (\ref{h18}) is void.  Similar situation occurs if    $(S,g)$ is complete and asymptotically flat  and $R\geq 0$ \cite{sc}. Condition $R\geq 0$ is necessarily  satisfied if  $g$ is  induced by a solution of the Einstein equations and $H$=0 on the initial  surface. This property follows from the Hamiltonian constraint which is one of the Einstein equations. For instance, let $g$  be the metric induced by the Kerr solution on the surface $t$=const. If $g$ together with a traceless tensor $K$,  different from that for the Kerr data, satisfy the momentum constraint one can apply theorem \ref{ma} without bothering about condition (\ref{h15}).

In what follows we show how to generate axially symmetric initial  data given a seed metric $g$ on $S$ such that the Sobolev inequality (\ref{h16}) is satisfied.
\begin{thm}\label{thm3}
Consider axially symmetric metric
\begin{equation}
g^{(u)}=\alpha^2 d\varphi^2+e^{2 u}g_{ab}dx^adx^b\label{h19}
\end{equation}
related to   a  complete and asymptotically flat metric $g=g^{(0)}$   of class $W_{\delta}^{k,p}$. Assume that   $(S,g)$ admits the Sobolev inequality (\ref{h16})  and that $u$ satisfies
\be\label{h22}
\|(R-2\tilde\Delta u)_-\|_{L^{3/2}}<\frac{8}{A}\ ,
\ee
where  the Ricci scalar $R$  and the norm refer to $g$ and $\tilde\Delta$ is the covariant Laplacian of metric   $g_{ab}dx^adx^b$.

If axialy symmetric data $(g^{(u)},K)$ with $H=0$ satisfy the momentum constraint and are asymptotically flat of class $W_{\delta}^{k,p}$ then there exist conformal data of the same class which satisfy all the constraint equations.
\end{thm}
\begin{proof}
Equation (\ref{eqn:scalar1}) corresponding to  (\ref{h19}) takes the form 
\begin{equation}\label{h20}
\bigtriangleup^{(u)}\psi=\frac 18 e^{-2u}(R-2\tilde\bigtriangleup u)\psi-\frac 18 K^{ij}K_{ij}\psi^{-7}
\end{equation}
with the covariant Laplacian $\bigtriangleup^{(u)}$ defined by $g^{(u)}$. 
 If $\psi$ doesn't depend on $\varphi$ equation (\ref{h20}) is equivalent to 
\begin{equation}\label{h20a}
\bigtriangleup\psi=\frac 18 (R-2\tilde\bigtriangleup u)\psi-\frac 18 e^{2u}K^{ij}K_{ij} \psi^{-7}\ ,
\end{equation}
where $\bigtriangleup$ is the  Laplacian corresponding to the seed metric $g$.
Let us consider  equation (\ref{h20a}) without assuming the axial symmetry of $\psi$. Inspection of  the proof of 
Theorem 1 in \cite{ma} shows that this theorem  is still valid if $R$ is replaced by another function from the same Sobolev space. Equation (\ref{h20a}) has a unique solution $\psi$ since (\ref{h22}) implies
\be\label{h21}
\lambda^{(u)}:=\operatorname{inf}_{f\in C^\infty_c(S)} \frac{\int_S(8|\nabla f|^2+(R-2\tilde\Delta u)f^2)dv_g}{\|f\|^2_{L^6}}>0\ ,
\ee
where the measure in the integral and the norm   correspond to  the metric $g$.
The function $\psi$   cannot depend on $\varphi$ since otherwise $\psi$ with shifted coordinate $\varphi$ would be another solution of equation (\ref{h20a}). It follows that $\psi$ satisfies also equation (\ref{h20}). In order to prove that there is no other solutions of this equation  let us make the conformal transformation (\ref{eqn:conf}). Then the Hamiltonian constraint (\ref{eqn:scalar}) is satisfied by  new (primed) fields, hence $R'\geq 0$. If $\psi$ was not the unique solution of (\ref{h20}) there must be a positive function $\xi$ (a ratio of two solutions of (\ref{h20})) which is not identically 1 but tends to 1 at infinity and  satisfies 
\begin{equation}\label{h20b}
\bigtriangleup'\xi=\frac 18 R'(\xi-\frac{1}{\xi^{7}})\ .
\end{equation}
If we multiply (\ref{h20b})  by $(\xi-\xi^{-7})$ and integrate it over $S$ we obtain
\begin{equation}\label{h20c}
\int_S{(1+\frac {7}{\xi^8})(\nabla'\xi)^2dv_g}+\frac 18 \int_S{R'(\xi-\frac{1}{\xi^{7}})^2dv_g}=0\ .
\end{equation}
Since both integrated expressions are nonnegative they have to vanish. Taking into account the asymptotic behavior of $\xi$ we obtain $\xi=1$ everywhere on $S$. Thus,  $\psi$ defined originally as a solution of (\ref{h20a}) is also a unique solution of equation 
(\ref{h20}).  

\end{proof}

Note that condition (\ref{h22}) is satisfied if
\be\label{h22a}
\|(\tilde\Delta u)_+\|_{L^{3/2}}<\frac{4}{A}\ .
\ee
If the l. h. s. of (\ref{h22a}) is finite one can achieve (\ref{h22a}) via transformation $u\rightarrow Cu$ with a suitably small value of constant $C$. If $u=0$ condition (\ref{h22a}) is trivially satisfied and the only problem is to solve the momentum constraint with respect to $K$ (see Proposition \ref{p1}). For instance, one can take the Kerr initial metric (\ref{8a}) and $K$ given by (\ref{28}) and (\ref{45}) with any $\omega$  satisfying asymptotic  condition (\ref{h11}).

\section{Hamiltonian constraint and  horizons}\label{ham1}
From the point of view of the black hole theory it is important  that initial data admit a surface  $S_0$ which can be considered as a black hole horizon.  Conditions on $S_0$ are usually formulated  in terms of functions  $\theta_{\pm}$ defined on $S_0$ by
\be
\theta_{\pm}=H-K(n,n)\pm 2h\ ,\ \ 2h=n^i_{\ |i}\ ,\label{-3}
\ee
where $n$ is the unit normal of $S_0$ oriented outside $S_0$. From the 4-dimensional point of view, $\theta_{\pm}$ are 
expansions of null geodesics emerging from  $S_0$ in the direction $k\pm n$, respectively. In order to interpret $S_0$ as a black hole horizon condition  $\theta_+=0$ is commonly assumed. Then $S_0$ is called marginally outer trapped surface (MOTS). Unfortunately, in general, Theorem \ref{ma} does not allow to control existence of MOTS or other trapped surfaces for final initial data, even  if data $(g,K)$ admit such a surface.

There is an exception to this rule  if $g$  has a  $Z_2$ symmetry preserving $S_0$ and if equation
\be\label{h30}
K_{ij}n^in^j=0
\ee
is satisfied on $S_0$. Then, from the uniqueness   assured by Theorem \ref{ma} solution  $\psi$ is also $Z_2$  symmetric, hence its normal derivative vanishes on $S_0$. All components of the  exterior curvature of $S_0$ embedded in $S$  vanish for both metrics  $g$ and $g'$. Hence,   $\theta_+=\theta_-=0$ for the ultimate initial data.

 An example of this type is again provided by the Kerr metric. In this case the surface $t=$const crosses the bifurcation surface (the Einstein-Rosen bridge) and on its other side  the Boyer-Lindquist radial coordinate $r$ grows again up to infinity. Thus, $r$  is not a global coordinate on $S$. A better coordinate  $\tilde r$  is related to $r$ by
\be\label{h24a}
r=M+\sqrt{M^2-a^2}\cosh{\tilde r}\ ,\ \tilde r\in [-\infty,\infty]\ .
\ee
Metric (\ref{8a}) is symmetric with respect to $\tilde r\rightarrow -\tilde r$ and the  exterior curvature of the surface $\tilde r=0$ vanishes. One can modify this metric according to Theorem \ref{thm3} with  $u$ and $\omega$ being even functions of $\tilde r$.   The corresponding conformal factor $\psi$ will be also $Z_2$ symmetric and the surface $\tilde r=0$ will  have  vanishing  null expansions $\theta_{\pm}$ with respect tvo the conformally transformed  data.

A construction of  initial data with MOTS is given in \cite{ma}. In this approach $S$ has an inner boundary $S_0$ and equation (\ref{eqn:scalar1}) is supplemented by the boundary condition                  
\be\label{h32}
n^i\p_i \psi+\frac 12 h\psi-\frac 14 K(n,n) \psi^{-3}=0
\ee
(we recall that we use fields $K_{ij}$ and $h$ which differ by sign from fields $K_{ij}$ and $h$  in \cite{ma}). 
 Condition (\ref{h32}) guarantees that  $\theta_+=0$   upon the conformal transformation (\ref{eqn:conf}). Unfortunately, the key existence theorem in \cite{ma} (Theorem  1) refers to properties of   conformally equivalent data  satisfying $R=0$. Since  these data are not known explicitly, it is highly nontrivial to satisfy assumptions of this theorem unless condition (\ref{h30}) is satisfied. In the latter case applying Theorem 1 from \cite{ma} to the data obtained by means 
of Corollary 1 from \cite{ma} yields  

\begin{thm}[Maxwell]\label{ma1}
 Let ($S,g$) be a  Riemannian manifold with an inner boundary $S_0$ and $(g,K)$ be a traceless solution of the momentum constraint which is asymptotically flat  of  class $W_{\delta}^{k,p}$. If $K(n,n)=0$ and
\be\label{h24}
\lambda_g:=\operatorname{inf}_{f\in C^\infty_c} \frac{\int_S(8|\nabla f|^2+R^{}f^2)dv_g-\int_{S_0} hf^2ds_g}{\|f\|^2_{L^6}}>0
\ee
then  equation (\ref{eqn:scalar1}) with the boundary condition 
\be\label{h33}
n^i\p_i \psi+\frac 12 h\psi=0
\ee
 possesses a  solution  $\psi>0$.  The conformally transformed data (\ref{eqn:conf}) satisfy all the  constraint equations and are asymptotically flat of class  $W_{\delta}^{k,p}$. The boundary $S_0$ is a marginally outer trapped surface  with $\theta_+=\theta_-=0$.
 \end{thm}
As in preceding section we are going to replace condition (\ref{h24}) by a  simpler one  under the assumption that  the  Sobolev inequality  (\ref{h16}) is satisfied.  For instance, it follows from Proposition \ref{prop0} that  $\lambda_g>0$ if $h\leq 0$ and  (\ref{h18}) is satisfied. In order to find less restrictive conditions let us introduce a compact Riemannian submanifold  $S'\subset S$ with a boundary containing $S_0$. The following identities are satisfied on  $S'$ \cite{a,Hebey}
\be
||f||_{L^2(S')}\leq B||f||_{L^6(S')}\label{60}
\ee
\be
||f||_{L^4(\p S')}\leq C ||\nabla f||_{L^2(S')}+D||f||_{L^2(S')}\ ,\label{61}
\ee
where $B$, $C$, $D$ are positive constants which depend on a choice of $S'$. 
\begin{prop}\label{prop4}
Let the flat Sobolev inequality (\ref{h16}) be satisfied on $S$.  Then  there is a constant $E$ such that inequality
\be
A^2||R_-||_{L^{\frac 32}(S)} +E^2||h_+||_{L^2(S_0)}\leq 8 \label{62}
\ee
implies (\ref{h24}).
\end{prop}
\begin{proof}
Since $||f||_{L^4(S_0)}\leq ||f||_{L^4(\p S')}$ and $||\cdot||_{L^p(S')}\leq ||\cdot||_{L^p(S)}$ it follows from (\ref{60}) and (\ref{61})  that 
\be
||f||_{L^4(S_0)} \leq C ||\nabla f||_{L^2(S)}+BD||f||_{L^6(S)}\label{63}
\ee
and consecutively
\be
||f||^2_{L^4(S_0)} \leq 2C^2 ||\nabla f||^2_{L^2(S)}+2B^2D^2||f||^2_{L^6(S)}\ .\label{63a}
\ee
Let us decompose $R$ and $h$ into positive and negative parts following (\ref{h16a}). 
 From $h\leq h_+$ and the H$\ddot{\texttt{o}}$lder inequality on $S_0$ one obtains
\be\label{64}
-\int_{S_0}{h f^2ds_g}\geq -||h_+||_{L^2(S_0)}||f||^2_{L^4(S_0)}\ .
\ee
It follows from (\ref{63a}), (\ref{64}) and (\ref{h17b}) that
\ba\label{67}
\int_S(8|\nabla f|^2+R^{}f^2)dv_g-\int_{S_0} hf^2ds_g\geq 
(8-2C^2||h_+||_{L^2(S_0)})||\nabla f||^2_{L^2(S)}\\\nonumber
-(\|R_-\|_{L^{\frac 32}(S)} +2B^2D^2
||h_+||_{L^2(S_0)})
||f||^2_{L^6(S)}\ .
\ea
Condition  (\ref{h24}) is satisfied if the r. h. s. of (\ref{67}) is greater or equal to $\lambda'||f||^2_{L^6(S)}$, where $\lambda'$ is a positive constant. The latter condition takes the form 
\be\label{65}
(8-2C^2||h_+||_{L^2(S_0)})||\nabla f||^2_{L^2(S)}\geq 
(\lambda'+\|R_-\|_{L^{\frac 32}(S)} +2B^2D^2
||h_+||_{L^2(S_0)})
||f||^2_{L^6(S)}\ .
\ee
The Sobolev inequality (\ref{h16}) implies (\ref{65}) with some $\lambda'>0$ provided that the norms of $R_-$ and $h_+$ satisfy (\ref{62}) with 
\be 
E^2=2C^2+2A^2B^2D^2\ .\label{68}
\ee

\end{proof}

Now, we will apply Theorem \ref{ma1} and Proposition \ref{prop4} to axially symmetric data from Section 2. If condition (\ref{45}) is satisfied
and  $S_0$   is axially symmetric  then the normal vector $n$ has no axial component and $K(n,n)=0$. Moreover one can choose  a conformal representative of metric such that  equation (\ref{eqn:scalar1}) takes the form characteristic for the flat metric. This property facilitates a possible  way to prove  the Sobolev inequality (\ref{h16}) and simplifies condition (\ref{62}).
\begin{thm}\label{thm5}
Let $S$ be a connected unbounded subset of the Euclidean space $R^3$ with an axially symmetric compact boundary $S_0$ such that the Sobolev inequality (\ref{h16}) is satisfied.  Let $u$ and $\omega$ be functions of  $r$ and $\theta$ and
\ba\label{h25}
g^{(u)}=r^2\sin^2{\theta}d\varphi^2+e^{2u}(dr^2+r^2d\theta^2)\\
K^{a3}=r^{-3}\eta^{ab}\omega _{,b}\sin^{-3}{\theta}\  ,\ \ K^{33}=K^{ab}=0\label{28h}
\ea
be   asymptotically flat data of  class $W_{\delta}^{k,p}$.
Let
\be
4A^2||(\tilde\bigtriangleup u)_+||_{L^{\frac 32}(S)} +E^2||h^{(0)}_+||_{L^2(S_0)}\leq 8 \label{69}
\ee
where $\tilde\Delta$ is the  Laplacian of the metric $dr^2+r^2d\theta^2$, the  norms correspond to flat  metric, $h^{(0)}$ is the mean curvature of $S_0$ with respect to flat metric and $E$ is given by (\ref{68}). 

Then there exist conformally equivalent data which satisfy all the  constraint equations and are asymptotically flat of class  $W_{\delta}^{k,p}$. The boundary $S_0$ is MOTS  with $\theta_+=\theta_-=0$.
\end{thm}
\begin{proof}
It is easy to show   that   equation (\ref{eqn:scalar1}) for data (\ref{h25})-(\ref{28h}) is equivalent to
\begin{equation}\label{h26}
\bigtriangleup\psi=-\frac 14 (\tilde\bigtriangleup u)\psi-\frac 18  r^{-4}\sin^{-4}{\theta}(\omega^{,a}\omega_{,a})\psi^{-7}
\end{equation} 
with the flat Laplacian $\Delta$. Let $\tilde n^i$ be the normal vector of $S_0$ with unit length with respect to flat metric $g^{(0)}$.  Boundary condition (\ref{h33}) takes the form
\be\label{70}
\tilde n^i\p_i \psi+\frac 12 h^{(0)}\psi=0\ ,\ \ 2h^{(0)}=\tilde n^i_{\ |i}
\ee
where $\tilde n^i_{\ |i}$ is defined by means of $g^{(0)}$. 
 From Theorem \ref{ma1} equation (\ref{h26}) and condition  (\ref{70}) have a unique solution $\psi$ if 
\be\label{h28}
\lambda^{(0)}_g:=\operatorname{inf}_{f\in C^\infty_c(S)} \frac{\int_S(8|\nabla f|^2-2f^2\tilde\Delta u)dv-\int_{S_0}h^{(0)}f^2ds}
{\|f\|^2_{L^6}}>0
\ee
where the integrals and the norm are defined by means of $g^{(0)}$. According to Proposition \ref{prop4}, assumption (\ref{69}) implies (\ref{h28}), so $\psi$ exists. 
Following the proof of Theorem \ref{thm3} one can show that  $\psi$  doesn't depend on $\varphi$ and it is also a unique solution of equation  (\ref{eqn:scalar1})  with condition (\ref{h33}), both corresponding to metric (\ref{h25}) (note that equation (\ref{h20c}) is still true since the normal derivative of $\psi'$ on $S_0$ has to vanish).

\end{proof}
Locally every metric (\ref{27}) can  be conformally transformed to the form  (\ref{h25}). 
 In the case of the initial Kerr metric one can simply  substitute (\ref{h24a}) and  $\tilde r=\ln{r'}$ into (\ref{8a}). Then a conformal transformation leads to (\ref{h25}) with $r'$ instead of $r$. The Kerr horizon corresponds to $r'=1$. In this case the initial surface $S$ is given by  $R^3$ with a removed  ball. We prove in Appendix that for such $S$ the Sobolev inequality (\ref{h16}) is satisfied.   Thus, Theorem \ref{thm5}  gives tools to  generalize the Kerr initial data.

A drawback of the approach  with an internal boundary is that, in general, we cannot control prolongation of initial data throughout  the boundary. Even if the metric $g$ before the conformal transformation  can be continued  to another asymptotically flat region it is not known whether  the conformal factor $\psi$ can be. This is because  Theorem \ref{ma1} can be applied to the exterior and interior regions independently but it  says nothing about values of $\psi$ on the boundary surface.

 If $S_0$ is the 2-dimensional  sphere a particular continuation is provided by the Bowen-York puncture method \cite{BY_80}. Let us consider   metric (\ref{h25}) with boundary at $r=1$.  If 
\be\label{29}
u_{,r}=0\ ,\ \psi_{,r}+\frac 12 \psi=0
\ee
at $r=1$, 
then this  boundary has the vanishing exterior curvature tensor corresponding to the final metric
\be
g'=\psi^4 g\ .
\ee
One can continue $g'$ through the surface $r=1$ putting $g'(1/r)=g'(r)$. An equivalent method is first to make the coordinate transformation $r=\exp{\tilde r}$ and then to assume that $g'(-\tilde r)=g'(\tilde r)$. We can complete so defined metric by the exterior curvature (\ref{28h}) with $\omega$ being an even function of $\tilde r$. In this way one
 obtains initial data with 2 asymptotically flat ends. These data are also available by use of Theorem \ref{thm3} with  
$g$ given by e.g. the initial Schwarzschild metric, but then  the Sobolev inequality (\ref{h16}) is more difficult to prove.

 \section{Summary}\label{sec:summary}
We have been studying  solutions of the vacuum  constraints in general relativity   such that, in general,  the  initial metric $g$
is not conformally flat and the mean exterior curvature is not constant. Section 2 concerns with the momentum constraint for  data with a continuous  symmetry. If the length $\alpha$ of the symmetry vector and function $\alpha_{,b}\alpha^{,b}$ are independent then  all solutions are given explicitly (Proposition 2.2). In other cases solutions are given partly explicitly and partly in terms of integrals  (Propositions 2.1, 2.2 and 2.3). Several simple families of solutions are presented.   For instance, condition (\ref{45}) defines a  class of solutions which contains data for  stationary axially symmetric metrics and also  for nonstationary solutions. Among them there are solutions which are asymptotically flat. 

Data without symmetries are investigated in Section 3 under  assumption (\ref{1}) about algebraic structure of the exterior curvature tensor. Special solutions of the momentum constraint are described by Propositions \ref{pr1} and \ref{pr2}. Among them there are asymptotically flat data with $H=0$ (see Corollaries \ref{cor1} and \ref{cor2} and hereafter).  These solutions are nonsymmetric generalization of the class of axially symmetric data which contains the Kerr initial data.

In order to prove solvability of the Hamiltonian constraint for asymptotically flat data (Section 4) we assume $H=0$ and  use the results of Maxwell \cite{ma} on the conformal method of Lichnerowicz, Choquet-Bruhat and York. We show that the crucial inequality (\ref{h15}) follows from a simpler one if the flat Sobolev inequality is satisfied (Proposition \ref{prop0}). More definite results are obtained if data are axially symmetric (Theorem \ref{thm3}). 

In order to encode marginally trapped surfaces into initial data (Section 5) we follow again the approach of Maxwell. Now the initial surface has an inner boundary which is supposed to become a marginally outer trapped surface after the conformal transformation solving the Hamiltonian constraint. We show again that the most important condition (\ref{h24})  follows from a simpler one (Proposition \ref{prop4}) and we present a  version of the existence theorem for a particular class of axially symmetric data (Theorem \ref{thm5}).

\null

\noindent
{\bf Acknowledgments.}
We are grateful to Niall Murchadha for turning our attention to references \cite{bp,d,cm}. 
This work is partially supported by the grant N N202 104838 of Ministry of
Science and Higher Education of Poland. 
\section*{Appendix}
It is known \cite{Hebey} that 
\begin{equation}
|u|_{L^6(R^3)}\leq A|\nabla u|_{L^2(R^3)}\label{118}
\end{equation}
for every $u\in C^1_c(R^3)$. 
Let  $M=R^3\setminus B$, where $B=B(0,b)$ is an open ball of a radius $b$ with a center at 0.  Given  $u\in C_c^1(M)$  and a parameter $\alpha\geq 1$ we  define the following function $u_\alpha$ in the ball
$$u_\alpha(r,\theta,\phi)=u(r^{-\alpha}b^{\alpha+1},\theta,\phi)\ ,$$ 
where $r$, $\theta$ and $\varphi$ are the spherical coordinates of $R^3$. In order to obtain inequality of type (\ref{118}) in $M$   we first prove the following estimation.
\begin{lem}\label{lem:reflection}
\begin{equation}\label{eqn:sobolev_lemma}
\|\nabla u_\alpha\|^2_{L^2(B)}\leq\alpha\|\nabla u\|^2_{L^2(M)}.
\end{equation}
\end{lem}
\begin{proof} The reasoning is purely computational.  Denote by $\dd^2\Omega$ the standard volume form on the unit  sphere . Now 
\begin{align*}
\|\nabla u_\alpha\|^2_{L^2(B)}=\int_0^b\dd r\int_{S^2}r^2\dd^2\Omega\left(\left|\frac{\pa u_\alpha}{\pa r}\right|^2+\frac 1{r^2}\left|\frac{\pa u_\alpha}{\pa\theta}\right|^2+\frac 1{r^2\sin^2\theta}\left|\frac{\pa u_\alpha}{\pa \phi}\right|^2\right)=\\
\int_0^b\dd r\int_{S^2}r^2\dd^2\Omega\left(\left(\frac{\pa u(s,\theta,\phi)}{\pa s}\frac{\pa s}{\pa r}\right)^2+\frac 1{r^2}\left(\frac{\pa u(s,\theta,\phi)}{\pa\theta}\right)^2+\frac 1{r^2\sin^2\theta}\left(\frac{\pa u(s,\theta,\phi)}{\pa \phi}\right)^2\right),
\end{align*}
where $s=b(\frac br)^{\alpha}$.

Changing the variables $r\mapsto s$ we get
$$r=b(\frac bs)^{1/\alpha}\ ,\qquad\frac{\pa s}{\pa r}=(-\alpha)(\frac br)^{\alpha+1}=(-\alpha)(\frac sb)^{1+1/\alpha}\ ,\quad \dd r=-\frac{1}{\alpha}(\frac bs)^{1+1/\alpha}\dd s\ .$$
Consequently
\begin{align*}
&\|\nabla u_\alpha\|^2_{L^2(B)}=\\
&\int_b^\infty\dd s\frac 1\alpha(\frac bs)^{1+1/\alpha}\int_{S^2}\dd^2\Omega\left(\left|\frac{\pa u}{\pa s}\right|^2\alpha^2 s^2+\left|\frac{\pa u}{\pa\theta}\right|^2+\frac 1{\sin^2\theta}\left|\frac{\pa u}{\pa \phi}\right|^2\right)=\\
&\alpha \int_b^\infty\dd s(\frac bs)^{1+1/\alpha}\int_{S^2}s^2\dd^2\Omega\left(\left|\frac{\pa u}{\pa s}\right|^2+\frac{1}{\alpha^2 s^2}\left|\frac{\pa u}{\pa\theta}\right|^2+\frac{1}{\alpha^2 s^2\sin^2\theta}\left|\frac{\pa u}{\pa \phi}\right|^2\right).
\end{align*}
Since $s\geq b$ and $\alpha\geq 1$, the latter expression is not greater than
$$\alpha\int_b^\infty\dd s\int_{S^2}s^2\dd^2\Omega\left(\left|\frac{\pa u}{\pa s}\right|^2+\frac{1}{ s^2}\left|\frac{\pa u}{\pa\theta}\right|^2+\frac{1}{s^2\sin^2\theta}\left|\frac{\pa u}{\pa \phi}\right|^2\right)=\alpha\|\nabla u\|^2_{L^2(M)}.$$
\end{proof}
\begin{thm}\label{thm:sobolev}
For $u\in C^1_c(M)$, where  $M=R^3\setminus B(0,1)$, the following inequality holds
$$\|u\|_{L^6(M)}\leq 2(2+\sqrt{2})A\|\nabla u\|_{L^2(M)},$$
where $A$ is the constant in the Sobolev inequality (\ref{118}) in $R^3$.  
\end{thm}
\begin{proof}
Let us consider the following prolongation $\wt u$ of $u$:
$$\wt u(r,\theta,\phi)=
\begin{cases}
u(r,\theta,\phi)\quad\text{for $r\geq b$}\\
3u_1(r,\theta,\phi)-2u_2(r,\theta,\phi)\quad\text{for $r< b$}
\end{cases}.$$
It is easy to check that $\wt u$ is a well defined function of class $C^1_c(R^3)$ such that $\wt u\big|_M=u$. Moreover, 
\begin{align*}
\|\nabla \wt u\|_{L^2(\R^3)}=&\|\nabla \wt u\|_{L^2(M)}+\|\nabla \wt u\|_{L^2(B)}=\\
&\|\nabla u\|_{L^2(M)}+\|3\nabla u_1-2\nabla u_2\|_{L^2(B)}\leq\\
&\|\nabla u\|_{L^2(M)}+3\|\nabla u_1\|_{L^2(B)}+2\|\nabla u_2\|_{L^2(B)}\leq\\
& 2(2+\sqrt{2})\|\nabla u\|_{L^2(M)}.
\end{align*}
where  Lemma 5.1 was used in the last  estimation. From this and (\ref{118}) one obtains
\be
\|u\|_{L^6(M)}\leq\|\wt u\|_{L^6(R^3)}\leq A\|\nabla \wt u\|_{L^2(R^3)}\leq 2(2+\sqrt{2})A\|\nabla u\|_{L^2(M)}\ .
\ee
\end{proof}


\begin{thebibliography}{99}
\bibitem{adm} 
Arnowitt R., Deser S. and  Misner C. W.: The dynamics of general relativity, in L. Witten, editor, \textit{Gravitation: An introduction to current research},  227--265, John Wiley, 1962.
\bibitem{York_73} York Jr. J W 1973 Conformally invariant orthogonal decomposition of symmetric
tensor on Riemannian manifolds and the initial-value problem of
general relativity \textit{J. Math. Phys.} \textbf{14}  456--64

\bibitem{Br_Lindq_63} Brill D R and Lindquist R W 1963
Interaction energy in geometromechanics
\textit{Phys. Rev.} \textbf{131} 471--6.

\bibitem{BY_80} Bowen J M and York Jr. J W 1980 Time-asymmetric initial data for black holes and black hole collisions \textit{Phys. Rev}. D \textbf{21} 2047--55

\bibitem{Br-Br_97} Brandt S and Brugmann B 1997 
A simple construction of initial data for multiple black holes
\textit{Phys. Rev. Lett.} \textbf{78} 3606--9

\bibitem{ChB_York_79}
Choquet-Bruhat Y and York Jr. J W 1980 The Cauchy problem  \textit{General Relativity and Gravitation} ed. A Held (New York: Plenum Press) 99-172 

 \bibitem{Isenberg_95} Isenberg J 1995
Constant mean curvature solutions of the Einstein constraint equations on closed manifolds \textit{Class. Quantum Grav.} \textbf{12}  2249--74

\bibitem{Br_Cant_81} Cantor M and Brill D R 1981
The Laplacian on asymptotically flat manifolds and the specification of scalar curvature
\textit{Compos. Math.} \textbf{43}  317--25

\bibitem{Cantor_77} Cantor M 1977
The existence of non-trivial asymptotically flat initial data for vacuum spacetimes \textit{Commun. Math. Phys.} \textbf{57} 83--96

\bibitem{ma} Maxwell D 2005 Solutions of the Einstein constraint equations with apparent horizon boundaries
\textit{Commun. Math. Phys.} \textbf{253} 561-583

\bibitem{And_Chr_96} Anderson A and  Chru\'{s}ciel P T 1996 On asymptotic behavior of solutions of the constraints equations in general relativity with ''hyperboloidal boundary conditions'' \textit{Dissert. Math.} \textbf{355}  1--100

\bibitem{And_Chr_Fri_92}  Anderson A, Chru\'{s}ciel P T and Friedrich H 1992
On the regularity of solutions to the Yamabe equation and the existence of smooth hyperboloidal initial data for Einstein's field equations \textit{Commun. Math. Phys.} \textbf{149} 587--612

\bibitem{Dain_04} Dain S 2004 Trapped surfaces as boundaries for the constraint equations \textit{Class. Quantum Grav.} \textbf{21} 555--71

\bibitem{Bartnik_Is_04} Bartnik R and  Isenberg J 2004 The constraint equations \textit{The Einstein Equations and the Large Scale Behavior of Gravitational Fields: 50 years of the Cauchy Problem in General Relativity} eds P T Chru\'{s}ciel and H Friedrich (Berlin: Birkh\"{a}user)

\bibitem{bp}  Baker J and  Puzio R  1999 New Method for Solving the Initial Value Problem with Application to
Multiple Black Holes \textit{Phys. Rev.} D \textbf{59} 04 4030

\bibitem{d} Dain S 2001 Initial Data for a Head-On Collision of Two Kerr-like Black Holes with Close Limit
\textit{Phys. Rev.} D \textbf{64} 12 4002

\bibitem{Gar_Pri_00} Garat A and Price R H 2000 Nonexistence of conformally flat slices of the Kerr spacetime \textit{Phys. Rev}. D  \textbf{61} 124011

\bibitem{cm} Conboye R and Murchadha N O
Potentials for transverse trace-free tensors arXiv:1306.1363

\bibitem{k} Stephani H, Kramer D, MacCallum M, Hoensalers C and Herlt E 2003 \textit{Exact Solutions of Einstein's Field Equations Second Edition} Cambridge University Press

\bibitem{b} Bartnik R 1986 The Mass of an Asymptotically Flat Manifold
\textit{Commun. Pure Appl. Math.} \textbf{39 }661-693 

\bibitem{sc} Salof-Coste L 2009 Sobolev Inequalities in Familiar and
Unfamiliar Settings \textit{Sobolev spaces in mathematics I} 299-343 Int. Math. Ser. 8 (Springer New York)

\bibitem{a} Adams R A, Fournier J F 2003 \textit{Sobolev Spaces } (Second ed.) Academic Press 

\bibitem{Hebey} Hebey E 2000 \textit{Nonlinear Analysis on Manifolds: Sobolev Spaces and Inequalities} Courant Lecture Notes in Mathematics 5, American Mathematical Society and Courant Institute of Mathematical Sciences





\end{thebibliography}
\end{document}